\documentclass{article}

\usepackage[utf8]{inputenc}
\usepackage[OT4]{fontenc}
\usepackage{amsmath}
\usepackage{amsfonts}
\usepackage{amssymb}
\usepackage{authblk}
\usepackage{subfigure}
\usepackage{bbold}
\usepackage[textwidth=1.2in]{todonotes}

\usepackage{tikz}
\usetikzlibrary{arrows,automata}
\usepackage{hyperref}
\usepackage{mathtools}
\usepackage{cite}
\usepackage{graphicx}
\usepackage{here}
\usepackage{todonotes}

\newtheorem{theorem}{Theorem}
\newtheorem{definition}{Definition}

\newtheorem{lemma}[theorem]{Lemma}

\newtheorem{conjecture}{Conjecture}

\newcommand{\ket}[1]{\ensuremath{|#1\rangle}}
\newcommand{\bra}[1]{\ensuremath{\langle#1|}}
\newcommand{\ketbra}[2]{\ensuremath{\ket{#1}\bra{#2}}}
\newcommand{\proj}[1]{\ensuremath{\ket{#1}\bra{#1}}}

\newcommand{\C}{\mathbb{C}}

\newcommand{\Tr}{\mathrm{Tr}}
\newcommand{\Ln}{\mathrm{L}}
\newcommand{\Cn}{\mathrm{C}}
\newcommand{\Tn}{\mathrm{T}}

\newcommand{\1}{{\rm 1\hspace{-0.9mm}l}}
\newcommand{\Id}{\1}
\newcommand{\kr}{\otimes}
\newcommand{\XX}{\mathcal{X}}
\newcommand{\YY}{\mathcal{Y}}

\newcommand{\RR}{\mathcal{R}}
\newcommand{\dd}{\mathrm{d}}
\newcommand{\diag}{\mathrm{diag}}
\newenvironment{proof}[1][Proof]{\addvspace{\baselineskip}\noindent\textbf{#1.}
}{\ \hfill\rule{0.5em}{0.5em}\par\addvspace{\baselineskip}}

\title{Relating entropies of quantum channels}
\author[1]{Dariusz Kurzyk}
\author[1]{{\L}ukasz Pawela\footnote{lpawela@iitis.pl}}
\author[1, 2]{Zbigniew Pucha{\l}a}

\affil[1]{Institute of Theoretical and Applied Informatics,
Polish Academy of Sciences, ul. Ba{\l}tycka 5, 44-100 Gliwice, Poland}
\affil[2]{Faculty of Physics, Astronomy and Applied Computer Science,
Jagiellonian University, ul. {\L}ojasiewicza 11, 30-348 Krak{\'o}w, Poland}

\begin{document}
\maketitle
\begin{abstract}
In this work, we study two different approaches to defining the 
entropy of a quantum channel. One of these is based on the von 
Neumann entropy of the corresponding Choi-Jamio{\l}kowski state. The 
second one is based on the relative entropy of the output of the 
extended channel relative to the output of the extended completely 
depolarizing channel. This entropy then needs to be optimized over 
all possible input states. Our results first show that the former 
entropy provides an upper bound on the latter. Next, we show that for 
unital qubit channels, this bound is saturated. Finally, we 
conjecture and provide numerical intuitions that the bound can also 
be saturated for random channels as their dimension tends to infinity.

\end{abstract}
	
\section{Introduction} 

One of the important areas of quantum information theory refers to an 
entropic picture of quantum states and operations. It is well known 
that the entropic uncertainty principle can be applied in quantum key 
distribution 
protocols~\cite{devetak2005distillation,berta2010uncertainty} in 
order to quantify the performance of these protocols. Another 
possible area where such an approach prevails is resource 
theory~\cite{gour2018resource}. The entropic approach in the 
description of quantum states can also be useful in
studies of quantum phenomena such as correlations or
non-locality~\cite{guhne2004characterizing,oppenheim2010uncertainty,rastegin2016separability,enriquez2015minimal}.
Another essential aspect of quantum information theory is studying 
the time evolution of quantum systems interacting with the 
environment. Entropic characterization of quantum operations can be 
helpful in the investigation of decoherence induced by quantum 
channel~\cite{roga2011entropic,roga2013entropic}. There exists
also numerous approaches to the formulation of entropic uncertainty
principles~\cite{rudnicki2014strong,coles2014improved,rastegin2016majorization,kurzyk2018conditional,puchala2018majorization},
which can be useful in the analysis of quantum key distribution, 
quantum
communication or characterization of generalized measurements. In
\cite{roga2011entropic,roga2013entropic} entropy of quantum channel is defined 
as the entropy of the state corresponding to the channel by the Jamio{\l}kowski 
isomorphism.

In quantum information theory, the relative entropy
$D(\rho\|\sigma)=\Tr(\rho \log \rho - \rho \log \sigma)$ plays an 
important
role \cite{vedral2002role} and can be useful in quantifying the 
difference
between two quantum states. In terms of quantum distinguishability, 
relative
entropy can be interpreted as a distance between two quantum
states. Nevertheless, it is crucial to remember that it is not a 
metric as it does
not fulfill the triangle inequality. It can be noticed that quantum 
transformation $\Phi$ cannot increase of a distinguishability between 
quantum states $\rho$, $\sigma$ what can be written as 
$D(\rho\|\sigma)\geq
D\big(\Phi(\rho)\|\Phi(\sigma)\big)$. This fact, sometimes called the 
data processing inequality plays an important role in the context of 
hypothesis
testing \cite{yuan2018entropy}. The quantum relative entropy can also 
be useful in the quantification of quantum entanglement. In this 
context, the amount of the entanglement of a quantum state $\rho$ 
can be defined as an optimal
distinguishability of the state $\rho$ from separable states e.g.
$\min\limits_{\sigma \in \mathrm{SEP}} D(\rho||\sigma)$, where 
minimization
is performed over separable states. It can be noticed that the 
relative
entropy can also be used to define of von Neumann entropy of $\rho$ as
$S(\rho)=\log d - D(\rho\|\1/d)$. This definition shows that
the entropy of a quantum state is related to its distance from the 
maximally
mixed state. This approach to von Neumann entropy is useful to define 
the
entropy of quantum channels. Formally it was studied by Gour and
Wilde~\cite{gour2018entropy}, where relative entropy of quantum 
channels was
introduced. According to their approach, the von Neumann entropy of 
the quantum
channel is given by an optimized relative entropy of the output from 
an
extended channel relative to the output of the extended depolarizing 
channel. The optimization
is performed over all possible input states.
Moreover, there is also a possibility to define other information 
measures, e.g., conditional entropy or manual information, in terms 
of relative entropy.
Recently the relative entropy of quantum channel and its 
generalizations are used in resource theory~\cite{liu2019resource}, 
studies of quantum channels e.g.
distinguishability~\cite{katariya2021geometric}, quantum channel
discrimination~\cite{fang2020chain} or channel
capacity~\cite{leditzky2018approaches,fang2019geometric}.

\section{Preliminaries}
Let $\XX$, $\YY$ denote complex Euclidean spaces, let $\dim(\XX)$ 
denote the dimension of the space $\XX$ and $\Ln(\XX, \YY)$ denotes a 
set of linear operators from $\XX$ to $\YY$. For simplicity we will 
write $\Ln(\XX) \equiv \Ln(\XX, \XX)$. 
If $\rho\in\Ln(\XX)$ is Hermitian 
($\rho=\rho^\dagger$), positive
semi-definite ($\rho\geq 0$) and a trace-one ($\Tr{\rho}=1$) linear 
operator
then $\rho$ is called as density operator. To keep our expressions 
simple, we
will write the density operator corresponding to a pure states as 
lowercase
Greek letters $\phi \equiv \proj{\phi}$. In order to keep track of 
subspaces of
composite systems we will write $\rho_{AB} \in \Ln(\XX_A \otimes 
\XX_B)$.

The set of all such mapping $\Phi: \Ln(\XX) \to \Ln({\YY})$ will be 
denoted by
$\Tn(\XX,\YY)$ and for brevity, we will write 
$\Tn(\XX)\equiv\Tn(\XX,\YY)$. A
mapping that is completely positive and trace-preserving is called a 
quantum
channel. The set of all quantum channels will be denoted $\Cn(\XX, 
\YY)$. There
exists a well-known bijection between the sets $\Cn(\XX, \YY)$ and 
$\Ln(\YY
\otimes \XX)$, the Choi-Jamiołkowski isomorphism. It is given by the 
relation
\begin{equation}
D_\Phi = (\Phi \otimes \Id) (\phi^+)
\end{equation}
where $\ket{\phi^+}=\sum_i^d\ket{ii}$ and $D_\Phi$ is called the 
dynamical matrix or Choi
matrix. Normalized $D_\Phi$ known as Choi-Jamio{\l}kowski state and 
will be denoted as $J_\Phi=D_\Phi/d$.

The von Neumann entropy of $\rho\in\Ln(\XX)$ is defined by the 
following formula
\begin{equation}
H(\rho)=-\Tr \rho \log \rho,\label{eq:ve}
\end{equation}
similarly to the classical Shannon entropy. This equation can be 
rewritten using the notion of relative entropy, which is defined for 
states $\rho$ and $\sigma$ analogously to its classical
counterpart~\cite{umegaki1962conditional}
\begin{equation}
D(\rho\| \sigma)=\Tr \rho (\log\rho-\log\sigma).
\end{equation}
Here we use the convention that $D(\rho \| \sigma)$ is finite when
$\mathrm{supp}(\rho) \subseteq\mathrm{supp}(\sigma)$. Otherwise, we 
put
$D(\rho\vert\sigma)=\infty$. Thus, we can rewrite Eq.~\eqref{eq:ve} as
\begin{equation}
H(\rho)=\log \dim(\XX)-D(\rho \| \rho_*),
\end{equation}
where $\rho_* = \1 / \dim(\XX)$.

The definition of quantum relative entropy can be extended to the 
case of quantum channels in the following 
manner~\cite{gour2018entropy}
\begin{equation}
D(\Phi\|\Psi) = \sup_{\rho_{AR} \in \Ln(\XX_A \otimes \XX_R)} 
D\Big((\Phi\kr\Id)(\rho_{AR})\|
(\Psi\kr\Id)(\rho_{AR})\Big).\label{eq:rec}
\end{equation}
The state $\rho_{AR}$ can be chosen as a pure state and the space 
$\XX_R$ can be
isomorphic to $\XX_A$. Utilizing Eq.~\eqref{eq:rec} we get the 
following
definition of the entropy of a quantum channel
\begin{definition}[\!\!\cite{gour2018entropy}] 
	Let $\Phi\in\Cn(\mathcal{X}_A,\mathcal{X}_B)$. Then its entropy 
	$H(\Phi)$is defined as 
	\begin{equation}\label{ex:rel}
	H(\Phi)= \log\dim\XX_B-D(\Phi\vert\vert\RR),
	\end{equation}
	where $\RR \in \Cn(\XX_A, \XX_B)$ is the depolarizing channel 
	$\RR : \rho \mapsto  (\Tr \rho) \ \1 / \dim(\XX_B)$.
\end{definition}
The quantum entropy was also defined in the same matter in
\cite{yuan2018entropy}. However, there exists an earlier definition 
of entropy
of a quantum channel. In~\cite{roga2011entropic,roga2013entropic}, the quantum 
channels was
characterized by the \emph{map entropy}, which was defined as the 
entropy of
corresponding Jamiołkowski state. It reads
\begin{definition}[\!\!\cite{roga2011entropic}]
	Let $\Phi\in\Cn(\mathcal{X}_A,\mathcal{X}_B)$. Its entropy 
	$H(\Phi)$ is given by
	the entropy of the corresponding Choi-Jamio{\l}kowski state
	\begin{equation}
	H^K(\Phi) = H(J_\Phi).
	\end{equation}
\end{definition}
The above entropy achieves its minimal value of zero for any unitary 
channel
and the maximal value of $2 \log \dim \XX_B$ for the completely 
depolarizing
channel. Based on these two definition we arrive at the following
observation.

\begin{lemma}\label{lemma:ineq}
	Let $\Phi \in \Cn(\XX_A, \XX_B)$. The two possible definitions of 
	quantum
	channel entropy $H(\Phi)$ and $H^K(\Phi)$ fulfill the following 
	relation
	\begin{equation}
	H(\Phi) \leq H^K(\Phi) - \log \dim \XX_B
	\label{eq:channel_ineq}
	\end{equation}
\end{lemma}
\begin{proof}
	The proof follows from a direct inspection
	\begin{equation}
	H(\Phi) = \log\dim\XX_B - \sup_{\ket{\psi} \in \XX_A \otimes 
	\XX_R}D\left((\Phi \otimes \Id)(\psi) \| (\RR \otimes 
	\Id)(\psi)\right).
	\end{equation}
	Let us denote $\sigma_{BR} = \left( \Phi\otimes \Id \right) 
	(\psi)$. Now we note that $\left( \RR
	\otimes \Id\right) (\psi) = \1 / \dim(\XX_B) \otimes \Tr_A \psi$ 
	and we use the well
	known identity $\log(\1 \otimes \rho) = \1 \otimes \log\rho$ and 
	we have
	\begin{equation}
	H(\Phi) = \log \XX_B - \sup_{\ket{\phi} \in \XX_A \otimes \XX_R} 
	\Big( -H(\sigma_{BR}) -\Tr \sigma_{BR} 
	\big(\frac{\1}{\dim(\XX_B)} \otimes \log \Tr_A \psi \big) 
	\Big).~\label{eq:9}
	\end{equation}
	Finally we note that $\Tr \sigma_{BR} \big(\1 \otimes \log(\Tr_A
	\psi)\big) = \Tr \Tr_B \sigma_{BR} \log \Tr_A \psi$ and $\Tr_B
	\sigma_{BR} = \Tr_A \psi$. Putting this into Eq.~\eqref{eq:9} 
	along with
	the fact that $J_\Phi$ is normalized we get the desired result.
\end{proof}

The main focus of this work is to find instances that saturate the 
inequality in
Lemma~\ref{lemma:ineq}. We will mainly focus on the study of unital 
qubit channels.

\section{Quantum unital qubit channels}

In this section we will focus our attention on unital qubit channels, 
that is
$\Phi \in \Cn(\C^2)$ such that $\Phi(\1) = \1$. Our goal here is to 
show
that the supremum present in Eq.~\eqref{eq:rec} is achieved for the 
maximally
entangled state $\ket{\phi^+}$. This can be formally written as the 
following
theorem

\begin{theorem}\label{th:qubit-unital}
	Let $\Phi \in \Cn(\C^2)$, such that $\Phi(\1)=\1$. Then
	\begin{equation}
	\label{thm:02}
	H(\Phi) = H^K(\Phi) - \log 2.
	\end{equation}
\end{theorem}

The remainder of this section contains technical lemmas which 
combined give the
proof of Theorem~\ref{th:qubit-unital}.

A generic two-qubit state can be written as
\begin{equation}
\ket{\psi_{AR}} = U\otimes V (\sqrt{p}\ket{00} + \sqrt{1-p} \ket{11}),
\end{equation}
for some qubit unitary matrices $U,V$. Let us note that the quantum 
relative entropy is
unitarily invariant $D(\rho\|\sigma)=D(U\rho U^\dagger\|U\sigma 
U^\dagger)$.
Moreover, we use the same the fact that the Jamiołkowski
matrix of channel $\Phi_A(\rho) = \Phi(W \rho W^\dagger)$ has the 
same spectrum
as the Jamiołkowski matrix of channel $\Phi$, where $A$ is a unitary 
matrix.
Thus, we can skip the unitary
operations in our further investigations. We may perform the 
optimization taking
into account only the parameter
$p$ which quantifies the amount of entanglement between the input 
qubits. In
order to further simplify notation we will write
\begin{equation}
\ket{\psi_{AR}}=\sqrt{p}\ket{00} + \sqrt{1-p}\ket{11}
=|\sqrt{P}\rangle\rangle\label{eq:state}
\end{equation}
where $|X\rangle\rangle$ denotes the vectorization of the matrix $X$ 
and
$\sqrt{P} = \diag(\sqrt{p}, \sqrt{1-p})$ and we define
\begin{equation}\label{eq:state2}
\phi_P = |\sqrt{P}\rangle\rangle \langle \langle \sqrt{P}|
\end{equation}

In the next step we will check the symmetry of $D(\Phi\|\RR)$ with 
respect
to the parameter $p$. Hence, we formulate the first lemma.
\begin{lemma}
	Let $\Phi \in \Cn(\C^2)$ and let $\phi_P$ be a two-qubit state as
	in Eq.~\eqref{eq:state2}. Then $D(\Phi 
	\| \RR)$ is symmetric in the parameter $p$.
\end{lemma}

\begin{proof}
	Let us denote $Q=\diag(1-p, p)$. It can be checked that
	\begin{equation}
	\big((\sigma_y\kr\sigma_y)(\1\kr \sqrt{P})D_{\Phi}(\1\kr
	\sqrt{P})(\sigma_y\kr\sigma_y)\big)^* = (\1\kr \sqrt{Q} 
	)D_{\Phi}(\1\kr
	\sqrt{Q}).
	\end{equation}
	This observation combined with the fact
	\begin{equation}
	(\Phi\kr\1)(\phi_P) = (\1\kr \sqrt{P})D_{\Phi}(\1\kr \sqrt{P}),
	\end{equation}
	gives the symmetry of the entropy
	\begin{equation}
	H\Big((\Phi\kr\1)(\phi_P)\Big)=H\Big((\Phi\kr\1)(\phi_Q)\Big).
	\end{equation}
	As for the term  $\Tr \Big((\Phi\kr\1)(\phi_P)\log 
	(\RR\kr\1)(\phi_P)\Big)$ 
	observe that
	\begin{equation}
	\log \big((\RR\kr\1)(\phi_P)\big)=\1 \otimes \log(P/2).
	\end{equation}
	Finally
	\begin{equation}
	(\sigma_y\kr\sigma_y)\Big(\log \big(\frac{1}{2}\1\kr 
	P\big)\Big)(\sigma_y\kr\sigma_y)=\log \big(\frac{1}{2}\1\kr 
	Q\big).
	\end{equation}
	Combining all of these observations yields the lemma.
\end{proof}

Subsequently, we prove in next lemma the concavity $D(\Phi\|\RR)$ 
with respect to the parameter $p$.
\begin{lemma}
	Given a unital channel $\Phi \in \Cn(\C^2)$ and let $\phi_P$ be a 
	two-qubit
	state as in Eq.~\eqref{eq:state2}. Then the function $f(p) =
	D\Big((\Phi\kr\1)(\phi_P)\vert\vert 
	(\mathcal{R}\kr\1)(\phi_P)\Big)$ is concave.
\end{lemma}
\begin{proof}
	For the purpose of this proof let us denote $\rho(p) = (\Phi 
	\otimes
	\1)(\phi_P)$. Let us also denote
	\begin{equation}
	\begin{split}
	g(p) = & \Tr \rho(p) \log \rho(p),\\
	l(p) = & \Tr \rho(p) \log (\RR \otimes \1)(\phi_P).
	\end{split}
	\end{equation}
	A direct calculation shows that $l(p) = h(p) + \log2$, where 
	$h(p)$ is the 
	point entropy. From this it follows that
	\begin{equation}\label{eq:part1}
	\frac{\dd^2 l}{\dd p^2} = -\frac{1}{p(1-p)} < 0.
	\end{equation}
	For $g(p)$ we calculate
	\begin{equation}\label{eq:second-term}
	\begin{split}
	\frac{d g}{dp}=&\Tr\Big(\rho'(p)\log 
	\rho(p)\Big)+\Tr\Big(\rho(p)\frac{d}{dp}\Big(\log 
	\rho(p)\Big)\Big)\\
	=&\Tr\Big(\rho'(p)\log \rho(p)\Big)+\Tr\Big(\rho(p)\rho^{-1}(p) 
	\rho'(p)\Big).
	\end{split}
	\end{equation}
	Observing that $\rho'(p) = \sqrt{J_\Phi}(\1 \otimes 
	\rho(p))\sqrt{J_\Phi}$ 
	we see that the second term in Eq.~\eqref{eq:second-term} is 
	equal to zero. 
	Hence, we have
	\begin{equation}
	\frac{\dd^2 g}{\dd p^2} = \Tr\left( \rho'(p) \frac{\dd}{\dd p} 
	(\log 
	\rho(p)) \right).
	\end{equation}
	From Taylor expansion of derivative formulae for matrix logarithms
	\cite{haber2018} we get
	\begin{equation}
	\frac{\dd}{\dd t} \log \rho(p) = \int_0^1 
	\Big(s(\rho(p)-\1)+\1\Big)^{-1}\rho'(p) 
	\Big(s(\rho(t)-\1)+\1\Big)^{-1} \dd s.
	\end{equation}
	Thus,
	\begin{equation}
	\begin{split}
	\frac{\dd^2 
		g}{\dd 
		p^2}=&\Tr\Bigg(\rho'(p)\int_0^1\Big(s\big(\rho(p)-\1\big)+\1\Big)^{-1}
	\rho'(p) \Big(s\big(\rho(p)-\1\big)+\1\Big)^{-1}\dd s\Bigg)\\
	\leq&\Tr\Bigg(\rho'(p)^2\int_0^1\Big(s
	\big(\rho(p)-\1\big)+\1\Big)^{-2} \dd s \Bigg).
	\end{split}
	\end{equation}
Now we will focus on the last integral above,
\begin{equation}
\begin{split}
\int_0^1\Big(s \, \rho(p)+(1-s) \, \1\Big)^{-2} ds &= 
U\Bigg(	\int_0^1\Big(s \, \lambda(A)+(1-s) \, \1\Big)^{-2} ds 
\Bigg)U^\dagger = \\
&= 
U \lambda^{-1}(\rho)U^\dagger=\rho(p)^{-1},
	\end{split}
\end{equation}
where $\lambda(\rho)$ is a diagonal matrix with eigenvalues of 
$\rho$ on a diagonal and $U$ is a unitary matrix.
According to the above considerations
\begin{equation}
\begin{split}
\frac{d^2 g}{dp^2}\leq& \Tr \Big(\big(\rho'(p)\big)^2 
\rho^{-1}(p)\Big)\\
=&\Tr \sqrt{D_{\Phi}}\Big(\1\kr \rho(p) \Big) 
\sqrt{D_{\Phi}}\big(\sqrt{D_{\Phi}}\big)^{-1}\big(\1\kr
P\big)^{-1}\big(\sqrt{D_{\Phi}}\big)^{-1}\sqrt{D_{\Phi}}\Big(\1\kr
\rho(p)\Big)\sqrt{D_{\Phi}}\\
=&\Tr J(\1\kr \rho(p) P^{-1} \rho(p))=\Tr (\Tr_A J) \rho(p) 
P^{-1} \rho (p)
= \Tr P^{-1}=\frac{1}{p(1-p)}.
\end{split}
\end{equation}
Combining this with Eq~\eqref{eq:part1} we see that $f(p)$ is concave.
\end{proof}
Based on the above lemmas, it can be concluded that supremum in
$D(\Phi\|\RR)$, where $\Phi(\1)=\1$, is obtained for $p=\frac12$, 
which
indicates it is achieved for the maximally entangled state 
$\ket{\phi^+}$.
Thus, combination of the lemmas proves Theorem~\ref{th:qubit-unital}.
	
\section{Asymptotic case}
In this section, we will show that Eq.~\ref{eq:channel_ineq} is 
saturated in
the case of large system size. Firstly, let us denote $\dim(\XX_A) = 
\dim(\XX_B) = d$. 
Numerical investigations lead us to formulate the following 
conjecture.
\begin{conjecture}
	Let $\Phi \in \Cn(\XX)$ and $d = \dim(\XX)$. Then as $d \to 
	\infty$
	\begin{equation}
	H(\Phi) \simeq H^K(\Phi) - \log d \simeq \log d -\frac{1}{2} + 
	o(1),
	\end{equation}
	where $\Phi$ chosen randomly according to measures introduced 
	in~\cite{nechita2018almost}.
\end{conjecture}

To provide some intuition behind this conjecture, we first state 
a theorem which tells us about the distribution of eigenvalues of 
the output of an extended random quantum channel, when the input 
is also chosen randomly.

This is summarized as Theorem~\ref{th:free-prod}.

\begin{theorem}\label{th:free-prod}
	Let $\Phi$ be a random channel with Jamio\l{}kowski matrix 
	$D_\Phi$, we assume
	that the limiting distribution of eigenvalues of $D_\Phi$ is 
	given by $\mu$.
	Let $\phi$ be a random pure state with the limiting distribution 
	of Schmidt
	values given by $\nu$. We define
	\begin{equation}
	\sigma = (\Phi \otimes \Id)(\proj{\phi}),
	\end{equation}
	then the limiting distribution of eigenvalues of $\sigma$ is 
	given by $\mu \boxtimes \nu$.
\end{theorem}
\begin{proof}\label{proof:01}
	Note that 
	\begin{equation}
	\ket{\phi} = (W \otimes \1) \sum_{i} \sqrt{\lambda_i} \ket{i, i},
	\end{equation}
	where $W$ is a random unitary matrix and $\{\ket{i}\}$ is the 
	computational 
	basis.
	\begin{equation}
	\begin{split}
	\sigma=(\Phi \otimes \Id)(\proj{\phi}) 
	&= \sum_{ij} \sqrt{\lambda_i \lambda_j} (\Phi \otimes \Id) ((W 
	\otimes \1)  \ketbra{i,i}{j,j} (W^\dagger \otimes \1) ) \\
	&= \sum_{ij} \sqrt{\lambda_i \lambda_j} (\Phi_W \otimes \Id) ( 
	\ketbra{i,i}{j,j}),
	\end{split}
	\end{equation}
	where $\Phi_W(\rho) = \Phi(W \rho W^\dagger)$, note that 
	Jamio\l{}kowski matrix
	of channel $\Phi_W$ has the same spectrum that the 
	Jamio\l{}kowski matrix of
	channel $\Phi$. Next we write 
	\begin{equation}
	\begin{split}
	\sigma &= \sum_{ij} \sqrt{\lambda_i \lambda_j} (\Phi_W \otimes 
	\Id) ( \ketbra{i,i}{j,j}) \\
	&= \sum_{ij}  \sqrt{\lambda_i \lambda_j} \Phi_W(\ketbra{i}{j}) 
	\otimes \ketbra{i}{j}\\
	&=
	(\Id \otimes \diag(\sqrt{\lambda}) ) D_{\Phi_W} (\Id \otimes 
	\diag(\sqrt{\lambda}) ).
	\end{split}
	\end{equation}
	Now note, that the eigenvalues of $\sigma$ are the same as 
	eigenvalues of
	$D_{\Phi_W} (\Id \otimes \diag({\lambda}) )$, which gives the 
	result.
\end{proof}

Now, we have the following intuition behind our conjecture. Combining 
the
results 
from~\cite{zyczkowski2011generating,puchala2016distinguishability}
with~\cite{nechita2018almost,kukulski2021generating} we have for 
large $d$
and uniform distribution of channels
\begin{equation}
H^K(\Phi) = 2\log d - \frac12 + o(1).
\end{equation}

Next, we have the following result. Let $\ket{\phi}$ be a random pure 
state with
the Schmidt numbers chosen according to some measure $\nu$ and let 
$\ket{\phi}$
be free from $\Phi$. Then the output state has its spectrum given by 
the free
multiplicative convolution $\mu \boxtimes \nu$, where $\mu$ is the 
distribution
of eigenvalues of $D_\Phi$.

Let us consider following optimization target
\begin{equation}
D(\Phi\|\RR) = \sup_{\ket{\phi} \in \Ln(\XX_A \otimes \XX_R)} 
D\Big((\Phi\kr\Id)(\ketbra{\phi}{\phi})\|
(\RR\kr\Id)(\ketbra{\phi}{\phi})\Big),
\end{equation}
where $\ket{\phi} = (U \otimes V) \sum_{i} \sqrt{\lambda_i} \ket{i, 
i}$ for
some unitary matrices $U$, $V$. Note that optimization result is 
invariant
under local operations $U$, $V$ on $\ket{\phi}$, but it depends on
$\lambda_i$. It can be checked that
\begin{equation}
\sigma=(\Phi\kr\Id)(\ketbra{\phi}{\phi})=(\Id \otimes 
\diag(\sqrt{\lambda}) ) D_{\Phi} (\Id \otimes 
\diag(\sqrt{\lambda}) )\label{eq:sigma}
\end{equation}
and 
\begin{equation}
\gamma=(\RR\kr\Id)(\ketbra{\phi}{\phi})=\Id/d \otimes 
\diag(\lambda).\label{eq:gamma}
\end{equation}
Next consider $D(\Phi\|\RR)=\sup_{\ket{\phi} \in \Ln(\XX_A \otimes 
\XX_R)} \Tr\sigma \log \sigma - \Tr \sigma \log \gamma$. 
Moreover,
\begin{equation}
\begin{split}
\Tr \sigma \log \gamma & = \Tr D_\Phi \cdot \Id \otimes 
\diag\left(\lambda\log \frac{\lambda}{d}\right) \\
& =\Tr \Tr_A D_\Phi \cdot \diag\left(\lambda\log 
\frac{\lambda}{d}\right)=-H(\lambda)-\log d
\end{split}
\end{equation}
The above expression reaches minimum for uniform distributed 
$\lambda$ and
them is equal to $\Tr \sigma \log \gamma=-2\log d$. Since $\sigma$ has
spectrum given by $\mu \boxtimes \nu$, then
\begin{equation}
\Tr\sigma \log \sigma \simeq -H(\mu \boxtimes \nu),
\end{equation}
where $\boxtimes$ denotes the multiplicative free convolution of 
measures
$\mu$ and $\nu$~\cite{voiculescu1987multiplication}. Assuming maximal
entropy $H(\lambda) = \log(d)$ implies $\nu =\delta(1/d)$, which 
behaves
like in identity in the operation $\boxtimes$. Hence, we have
\begin{equation}
	H(\mu \boxtimes \nu) = H(\mu),
\end{equation}
which gives us
\begin{equation} 
	D(\Phi \| \RR) = \frac12.
\end{equation}
Now, going back to the entropy of the channel $\Phi$ we have
\begin{equation}
H(\Phi) \simeq \log(d) -  \frac12 +o(1).
\end{equation}

The intuition behind our assumption that $\mu = \delta(1/d)$ is 
presented in
Fig.~\ref{fig:intuition}. In it, we present the quantity $D(\sigma
\|\gamma)$ where $\sigma$ and $\gamma$ are as in 
Eqs.~\eqref{eq:sigma} and
\eqref{eq:gamma} respectively. The plots are presented for various
distributions $\nu$ of the Schmidt numbers of the input state 
$\ket{\psi}$.
The red line shows the case $\nu = \mathrm{Dir}(d, 1)$, the blue line 
shows
the case when $\nu = \mathrm{Dir}(2, 1)$, the yellow line is the case 
$\nu =
\mathrm{Dir}(d, 2)$ and finally, the green line shows the case $\nu =
\delta(1/d)$. The dashed line is the quantity $\log(d) - \frac12$. As 
can be
Subsequently the more non-zero Schmidt numbers and the more they are
concentrated in the center of the simplex $\Delta_{d-1}$, the closer 
we get
to the quantity we conjecture. Finally, when we choose a deterministic
distribution in the center of simplex, we achieve the optimal value.

\begin{figure}
	\centering\includegraphics{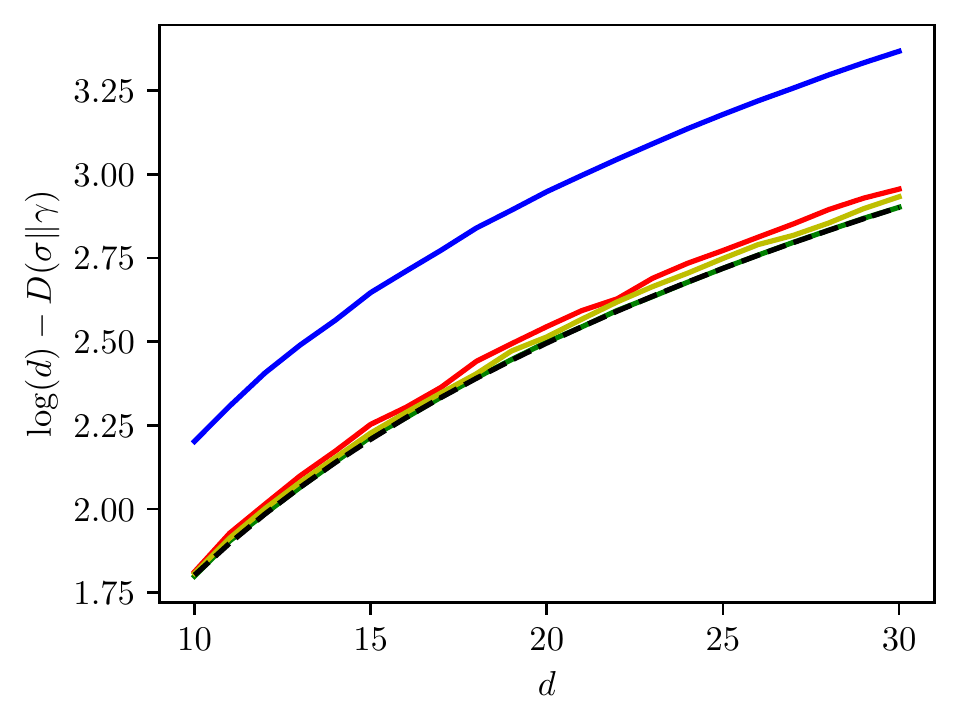} \caption{The quantity
	$D(\sigma \| \gamma)$ where $\sigma$ and $\gamma$ are as in
	Eqs.~\eqref{eq:sigma} and \eqref{eq:gamma} respectively. The 
	plots are
	for  $\nu = \mathrm{Dir}(d, 1)$ (red),  $\nu = \mathrm{Dir}(2, 1)$
	(blue), $\nu = \mathrm{Dir}(d, 2)$ (yellow) and $\nu = 
	\delta(1/d)$
	(green). The dashed line is the quantity $\log(d) -
	\frac12$.}\label{fig:intuition}
\end{figure}
\section{Conclusion}

In this paper, we discuss two approaches to entropic 
quantification of quantum channels. We begin our studies with a 
formulation of a lemma, which describes a relation between the 
entropy of quantum channels proposed by Gour and Wilde 
\cite{gour2018entropy} and entropy of Jamiołkowski matrix of 
quantum 	channels~\cite{roga2011entropic,roga2013entropic}. 
We show that both definitions give the same value up to an 
additive constant in the case of the quantum unital qubit 
channels. This part of our considerations uses the mathematical 
language of distinguishability of quantum states and channels. 
Therefore we assume that obtained results can be used to study 
resource theories and hypothesis testing.
We also provide a conjecture backed by numerical experiments that 
both formulas provide the same results up to an additive constant 
in the case of large system size.

\section*{Acknowledgements}
This work was supported by the Polish National Science Centre under grant
number 2016/22/E/ST6/00062.

\bibliography{bib} 
\bibliographystyle{ieeetr}
	
\end{document}